\documentclass[12pt,final,
thmsa,
epsf,a4paper]
{article}
\usepackage{booktabs}

\newtheorem{theorem}{Theorem}

\newtheorem{definition}{Definition}

\newenvironment{proof}[1][Proof]{\emph{#1.} }{\  \hfill $\square $ \vspace{5 pt}}

\usepackage{natbib}
\usepackage{amssymb}
\usepackage[dvips]{graphics}
\usepackage{amsmath}

\usepackage{mathpazo}
\usepackage{enumerate}

\usepackage{amsfonts}

\usepackage{amssymb}

\usepackage{enumerate}

\usepackage{mathrsfs}

\usepackage[usenames]{color}
\usepackage{cancel}

\usepackage{pgf,tikz}

\usepackage{multicol,multirow}

\usetikzlibrary{decorations.pathreplacing,angles,quotes}
\usetikzlibrary{arrows.meta}
\usetikzlibrary{arrows, decorations.markings}
\tikzset{myptr/.style={decoration={markings,mark=at position 1 with %
       {\arrow[scale=2,>=stealth]{>}}},postaction={decorate}}}

\usepackage{hyperref}
\hypersetup{
    colorlinks,
    citecolor=blue,
    filecolor=black,
    linkcolor=blue,
    urlcolor=blue
}

\usepackage{pgf,tikz}
\usetikzlibrary{arrows}

\usepackage[utf8]{inputenc}
\usepackage{amssymb, amsmath,geometry}
\usepackage{fancyhdr}

\usepackage{soul}

\usepackage{emptypage}

\newcommand*\samethanks[1][\value{footnote}]{\footnotemark[#1]}

\usepackage[T1]{fontenc}
\DeclareFontFamily{T1}{calligra}{}
\DeclareFontShape{T1}{calligra}{m}{n}{<->s*[1.44]callig15}{}
\DeclareMathAlphabet\mathcalligra   {T1}{calligra} {m} {n}

\begin{document}

\title{  Nash implementation in a many-to-one matching market
\thanks{We thank Agustín Bonifacio and Pablo Neme for their very detailed comments. We acknowledge financial support
from UNSL through grants 032016 and 030320, from Consejo Nacional
de Investigaciones Cient\'{\i}ficas y T\'{e}cnicas (CONICET) through the grant
PIP 112-200801-00655, and from Agencia Nacional de Promoción Cient\'ifica y Tecnológica through grant PICT 2017-2355.}}


\author{Noelia Juarez\thanks{Instituto de Matemática Aplicada San Luis (UNSL-CONICET) and Departamento de
Matemática, Universidad Nacional de San Luis, San Luis, Argentina. Emails: \texttt{nmjuarez@unsl.edu.ar} (N. Juarez), 
\texttt{pbmanasero@unsl.edu.ar} (P. B. Manasero) and \texttt{joviedo@unsl.edu.ar} (J. Oviedo).} 
 \and Paola B. Manasero\samethanks[2] \and Jorge Oviedo\samethanks[2]}

\date{\today}
\maketitle

\begin{abstract}
 In a many-to-one matching market, we analyze the matching game induced by a stable rule when firms' choice function satisfy substitutability. We show that any stable rule implements the individually rational correspondence in Nash equilibrium when both sides of the market play strategically. Moreover, when only workers play strategically and firms' choice functions satisfy the law of aggregate demand, we show that the firm-optimal stable rule implements the stable correspondence in Nash equilibrium.

\bigskip

\bigskip


\end{abstract}

\section{Introduction}

In this paper, we study a many-to-one matching market in which agents on one side of the market (which we call \emph{firms}) must be assigned to subsets of agents on the other side of the market (which we call \emph{workers}), and the only requirement on subsets of workers that each firm's choice function must satisfy is substitutability. \cite{kelso1982job} introduced this condition, which is the weakest requirement on firms' choice functions in order to guarantee the existence of stable matchings. A firm has substitutable choice functions if it wants to continue hiring a worker even if other workers become unavailable.  

In centralized markets, a board needs to collect the preferences and choice functions of all agents in order to produce a stable matching. Normally, agents are expected to behave strategically by not revealing their true preferences or their true choice functions in order to benefit themselves. When this is the case, the matching market becomes a matching game. 

A stable rule is a function that associates each stated strategy profile to a stable matching under this stated profile. To evaluate such matchings workers, and firms use their true preferences and their true choice functions, respectively. A market and a stable rule induce a matching game. In this game, the set of strategies for each worker is the set of all possible preferences that she could state. Similarly, the set of strategies for each firm is the set of all possible choice functions that it could state.

In this paper, the equilibrium concept we focus on is Nash equilibrium. In a Nash equilibrium, no agent improves by deviating from its initial chosen strategy, assuming the other agents keep their strategies unchanged.

 It is well known, for matching markets, that there is no stable rule for which truth-telling is a dominant strategy for all agents \citep[see][among others]{dubins1981machiavelli, roth1982economics, roth1985college, sotomayor1996admission, Sotomayor2012, martinez2004group, manaserooviedo2022}. That is, given the true preferences and a stable rule, at least one agent might benefit from misrepresenting her preferences regardless of what the rest of the agents state. Thus, stable matchings cannot be reached through  dominant ``truth-telling equilibria". The stability of equilibrium solutions under true preferences is expected to be affected when agents behave strategically.

Given a market, the question that arises is what are the rules that induce a matching game that allows us to implement stable matchings in Nash equilibrium. First, we study a matching game in which the players are all the agents. We show that any stable  rule implements, in Nash equilibrium, the individually rational matchings. Second, to implement stable matchings, we focus on another matching game. In some markets, such as school choice or labor markets, institutions are legally required to declare their true preferences (priorities or choice functions), i.e., their preferences are public. In these cases, individuals (students or workers) are expected to manipulate their preference lists to their advantage. This situation leads us to consider a matching game in which the players are only the workers.

 Furthermore, when firms’ choice functions satisfy, in addition to 
substitutability, the ``law of aggregate demand" (LAD, from now on),\footnote{This property is first studied by \citet{alkan2002class} under the name of ``cardinal monotonicity". See also \citet{hatfield2005matching}.} we show that the firm-optimal stable rule implements, in Nash equilibrium, the stable matchings.

The contribution of this paper is to generalize the approach presented by \citet{Sotomayor2008,Sotomayor2012} for the many-to-one matching market with responsive preferences (a more restrictive requirement than substitutability) to substitutable choice functions.
 
There is an extensive literature that focuses on studying the implementation of rules using Maskin's results as the main tool  \citep [see] [among others]{maskin77,maskin99,kara1996nash,kara1997implementation, ehlers2004monotonic, haake2009monotonicity}. In addition, to study implementation, the aforementioned authors analyze the relationship between stability, monotonicity, individual rationality, and Pareto efficiency.    \citet{maskin77,maskin99}  shows that a monotonicity condition (Maskin's monotonicity) is necessary for a rule to be Nash implementable. Also, he shows that Maskin's monotonicity and no veto power together are sufficient conditions for implementability. It is important to highlight that, unlike previous work, our results cannot be obtained through Maskin's implementation result since, although stable rules satisfy Maskin's monotonicity, they do not satisfy no veto power \citep [see] [for more details]{kara1996nash} .

For the one-to-one and  many-to-one markets with responsive preferences, \citet{kara1996nash,kara1997implementation} show that the stable rules are Nash implementable.
 \citet{ehlers2004monotonic} obtains positive implementation results in one-to-one markets when agents are allowed to have weak preferences. In a many-to-one market with contracts, \citet{haake2009monotonicity} show that stable rules are Maskin monotonic and implementable. All of the above articles demonstrate the implementability of stable rules using some implementation conditions; e.g. monotonicity \citet{maskin99}, essential monotonicity \citet{yamato} or the implementability condition of \citet{moorerepullo}. In contrast, we focus on studying a game and identifying the strategies that are  Nash equilibria of the game and that allow us to implement
stable solutions.
 
 In a one-to-one  market, \citet{ alcalde1996implementation} studies the design of specific mechanisms to implement stable solutions. He introduces two types of mechanisms. One implements the set of all stable matchings in undominated Nash equilibria; the other implements the
optimal stable matching for one of the two sides of the market via dominance resolvability. This last mechanism is the result of
the classic algorithm in matching theory, the Gale-Shapley mechanism.

Unlike us, \citet{ALCALDE2000294} analyze a two-stage game in a many-to-one market under substitutability. They focus on studying the notions of subgame perfect Nash equilibrium and strong subgame perfect Nash equilibrium. Under these notions of equilibrium, they implement the stable correspondence. Thus, their approach is tangential to ours. 

The rest of the paper is organized as follows. Section \ref{seccion de preliminares} presents the model and
some preliminaries. All the results of the paper are presented in Section \ref{S2}. First, we show that any stable rule implements the individually rational correspondence in Nash equilibrium. Second, assuming that only workers play strategically, we show that the firm-optimal stable rule implements, in Nash equilibrium, the stable correspondence when firms’ preferences
satisfy substitutability and LAD. Finally, concluding remarks are gathered in Section \ref{Concludings}.

\section{Model and preliminaries}\label{seccion de preliminares}
We consider a many-to-one matching market where there are two disjoint sets of agents: the set of \textit{workers} $W$ and the set of \textit{firms} $F$.  Each worker $w\in W$ has a strict preference relation $P_w$ over the individual firms and the prospect of being unmatched,  denoted by $\emptyset$.  For each $ w\in W $, $R_w$ is the weak preference over $ F $  associated with $P_w$. Each firm $f\in F$ has a choice function $C_f$ over the set of all subsets of $W$ that satisfies \textbf{substitutability}, i.e.,   for $S'\subseteq S \subseteq W$, we have $C_f(S)\cap S' \subseteq C_f(S')$.\footnote{Substitutability is equivalent to the following: for each $w\in W$ and each $S\subseteq W$ such that $w\in S$, $w\in C_f\left(S\right)$ implies that $w\in C_f\left(S'\cup \{w\}\right)$ for each $S' \subseteq S.$} In addition, we assume that $C_f$ satisfies $C_f(S')=C_f(S)$ whenever $C_f(S)\subseteq S' \subseteq S \subseteq W.$ This property is known in the literature as \textbf{consistency}. If $C_f$  satisfies substitutability and consistency, then it also satisfies \begin{equation}\label{propiedad de choice}
C_f\left(S\cup S'\right)=C_f\left(C_f\left(S\right)\cup S'\right)
\end{equation} for each pair of subsets $S$ and $S'$ of $W$.\footnote{This property is known in the literature as \emph{path independence} \citep[see][]{alkan2002class}.} 
Let $P=(P_w)_{w\in W}$ be the preference profile for all workers, and let $C=(C_f)_{f\in F}$ be the profile of choice functions for all firms. A \emph{many-to-one matching market} is denoted by $(W,F,P,C).$ Since the sets $F$ and $W$ are kept fixed throughout the paper, we often identify the market $(W,F,P,C)$ with the profile $(P,C)$.
Given the profile $P$, we consider that each $w\in W$ may
misrepresent her preferences $P _{w}$, by some 
preference $P'_w$.  We denote by $\left(P'_{w} ,P
_{-w}\right) $  the new profile with the misreport, where $P_{-w}$  is the subprofile obtained by removing  $P_{w}$  from $ P $ . We denote by $\mathcal{P}$  the set of all worker preference profiles.  We consider that each $f\in F$ may
misrepresent her choice function $C _{f}$, by some 
choice function $C' _{f}$.  We denote by $\mathcal{C}$  the set of all firms' choice functions profiles. Throughout this paper, we assume that $C$ is substitutable for each $ C\in \mathcal{C} $.


\begin{definition}
A \textbf{matching} $\mu$ is a function from $W\cup F$ into $2^{W\cup F}$ such that, for each $w\in W$ and  each $f\in F$:
\begin{enumerate}[(i)]
\item $\mu(w)\subseteq F$ with $|\mu(w)|\leq 1.$ 
\item $\mu(f)\subseteq W$.
\item $w\in \mu(f)$ if and only if $\mu(w)=\{f\}$.\footnote{Usually, we will omit the curly brackets. For instance, instead of condition (iii) we will write: ``$w\in \mu(f)$ if and only if $\mu(w)=f$''. No confusion will arise.} 
\end{enumerate}
\end{definition}

Let $ \mathcal{M} $ be the set of all matchings.
An agent $a\in W\cup  F$ is \textbf{matched} if $\mu(a) \neq \emptyset$, otherwise  $a$ is \textbf{unmatched}. A matching $\mu$ is \textbf{blocked by a worker $\boldsymbol{w}$} if $\emptyset P_w \mu(w)$; that is, worker $w$  would rather be unemployed than work for firm $\mu(w)$. Similarly, $\mu$ is \textbf{blocked by a firm $\boldsymbol{f}$} if $\mu(f)\neq C_f\left(\mu(f)\right)$; that is, firm $f$ wants to fire some workers in $\mu(f)$. A matching is \textbf{individually rational} if it is not blocked by any individual agent.  The set of individually rational matchings for market $(P,C)$ is denoted by $\boldsymbol{I(P,C)}.$  

A matching $\mu$ is \textbf{blocked by a firm-worker pair $\boldsymbol{(f,w)}$} if  $w \in C_f\left(\mu( f )\cup \{w\}\right),$  and $f P_w \mu( w )$; that is, if they are not matched by $\mu$, firm $f$ wants to hire $w$, and worker $w$ prefers firm $f$ to $\mu(w)$.  A matching $\mu$ is \textbf{stable} if it is individually rational and it is not blocked by any firm-worker pair. The set of stable matchings for market $(P,C)$ is denoted by $\boldsymbol{S(P,C)}.$ 

For each $ f\in F $, the choice function $ C_f $ induces a preference binary relation as follows: given two sets of workers $S,T\subseteq W$, we write  $\boldsymbol{S \succeq_{f} T}$ if and only if $S=C_f\left(S \cup T\right)$.\footnote{This relation $\succeq  $ was used in \citet{blair1988lattice}.} Furthermore, given two matchings $\mu$ and $\mu'$, we write    $\boldsymbol{\mu \succeq_{F} \mu'},$ whenever $\mu(f) \succeq_{f}\mu'(f)$ for each $f \in F.$ Similarly, we write $\boldsymbol{\mu R_{W} \mu'},$ whenever $\mu(w) R_w\mu'(w)$ for each $w \in W.$

The set of stable matchings under substitutable choice functions is very well structured. \citet{blair1988lattice} proves that this set has two lattice structures, one concerning $\succeq_F$ and the other one concerning $R_W$. Furthermore, it contains two distinctive matchings: the firm-optimal stable matching under $ (P,C) ,$ $\mu_F(P,C)$, and the worker-optimal stable matching under $ (P,C) $, $\mu_W(P,C)$. The matching $\mu_F$ is unanimously considered by all firms to be the best among all stable matchings and $\mu_W$ is unanimously considered by all workers to be the best among all stable matchings \cite[see][for more details]{roth1984evolution,blair1988lattice}.

In a seminal paper, \citet{gale1962college} introduce the Deferred Acceptance (\emph{DA}, from now on) algorithm, which computes the optimal stable matching for one side of the market. Later, the \emph{DA} algorithm is adapted to the many-to-one market when firms preferences are substitutable, by \citet{roth1992two}.

\section{Nash Implementation}\label{S2}
\subsection{The individually rational correspondence}

 Before introducing the matching game, we need some standard terms and notation.

A \textbf{social choice correspondence} $  \Psi: \mathcal{P}\times \mathcal{C}  \rightarrow 2^{\mathcal{M}}$ selects, for each market $ (P',C')\in \mathcal{P}\times \mathcal{C} $, a subset of matchings $ \Psi(P',C')\subseteq \mathcal{M}.  $ The \textbf{individually rational correspondence} $ \mathcal{I}:\mathcal{P}\times \mathcal{C}  \rightarrow 2^{\mathcal{M}} $ selects for each market $ (P',C')$, the set of individually rational matchings $ I(P',C') $. 
In this context, $P'\in \mathcal{P}$ is a strategy profile for all workers where for each $ w\in W $, $P'_w$ is  $ w $'s strategy. Moreover $ C'\in \mathcal{C} $ is a strategy profile for all firms, where for each $ f\in F $, $ C'_f $ is the $ f $'s strategy.  A \textbf{strategy profile} for all agents is $ (P',C')\in \mathcal{P}\times \mathcal{C}. $ A \textbf{(matching) rule} is a function  $\psi:  \mathcal{P}\times \mathcal{C} \rightarrow \mathcal{M} $  that selects for each strategy profile $ (P',C') \in \mathcal{P}\times \mathcal{C}  $ a matching $ \psi (P',C')\in \mathcal{M}. $  A matching rule $ \psi $ is \textbf{stable} if $ \psi(P',C')\in S(P',C') $ for each $ (P',C')\in \mathcal{P}\times \mathcal{C}.$ 

Given  $  ( P,C)\in \mathcal{P}\times \mathcal{C} $ and a stable rule $ \psi$ the \textbf{(matching) game} induced by $   (P,C) $ and $ \psi $ is denoted by $ (\psi, P,C). $ A strategy profile $(P',C')  $ is a \textbf{Nash equilibrium} of $ (\psi,P,C) $, if no agent can achieve a better outcome deviating from her strategy, assuming that the other agents do not deviate from the strategy profile $(P',C')  $.  Formally,

\begin{definition}
Let be $(\psi, P,C) $ the game induced by $ (P,C) $ and $ \psi $. A strategy profile $ (P',C') $ is a \textbf{Nash equilibrium} of $(\psi, P,C) $  if for each $ w\in W $, $\psi(P',C')(w)R_w\psi( \widehat{P}_w,P'_{-w},C')(w) $ for each strategy $ \widehat{P}_w  $ of $ w $, and for each $ f\in F $, $ \psi(P',C')(f)\succeq_f\psi(P',\widehat{C}_{f},C'_{-f})  (f)$ for each strategy  $\widehat{C}_f $ of $ f $.
\end{definition}

\begin{definition}\label{defimple}

 We say that the game $ (\psi,P,C) $ \textbf{implements} the social choice correspondence $ \Psi $ in Nash equilibrium if,
\begin{enumerate}[(i)]
\item     for each Nash equilibrium  $ (P',C') $ of $(\psi, P,C) $, $ \psi(P',C')\in \Psi(P,C) $, 

 \item  for each matching $ \mu \in \Psi(P,C) $ there is a Nash equilibrium $ (P',C') $ of the game $ (\psi,P,C) $ such that $ \psi(P',C')=\mu. $
 \end{enumerate}
 \end{definition}

Next, we show that any stable matching rule implements the individually rational correspondence in Nash equilibrium.\footnote{This result generalizes the result first presented by \citet{alcalde1996implementation} for the marriage market and then extended by \citet{Sotomayor2012} for the many-to-one matching market with responsive preferences. }

\begin{theorem}\label{ToremaIR}
Let $ (P',C')$ be a market and let  $ \psi:\mathcal{P}\times \mathcal{C}\rightarrow \mathcal{M} $ be a stable matching rule. Then, the game $ (\psi, P,C) $ implements the individually rational correspondence  $ \mathcal{I} $ in Nash equilibrium.
\end{theorem}

\begin{proof}
Let $ (P,C) $ be a market and $ \psi $ a stable matching rule. In order to show that the game $ (\psi,P,C) $ implements $ \mathcal{I} $ in Nash equilibrium, we need to prove items \textit{(i)} and \textit{(ii)} of Definition \ref{defimple}.
\begin{enumerate}
\item [$\boldsymbol{(i)}$] Let $ (P',C') $ be a Nash equilibrium of the game $(\psi,P,C)$. We prove that $ \psi(P',C')\in \mathcal{I}(P,C)$. Assume that $ \psi(P',C')$ is blocked by a worker $ w $, $ \emptyset P_w \psi(P',C')(w)$. Then the worker can improve by choosing the strategy in which no firm is acceptable i.e., $ P''_w=\emptyset$.  Thus, $ \psi(P''_w,P'_{-w},C')(w)P_w\psi(P',C')(w)$ contradicting that $ (P',C') $ is a Nash equilibrium of the game. Hence, $ \psi(P',C') $ is not blocked by any worker under $ (P,C) $. Now, we prove that $ \psi(P',C')$ is not blocked by any firm. Let $ f \in F$ since $ (P',C') $ is a Nash equilibrium, for each strategy $ C''_{f} $ we have that 
\begin{equation}\label{defENF}
 \psi(P',C')(f)= C_f\left(\psi(P',C')(f)\cup \psi(P',C''_f,C'_{-f})(f) \right) .  
\end{equation}
Choose a choice function  $ C''_{f} $ such that $C''_{f}(W)=\emptyset $. Since $ \psi $ is a stable matching rule, $ \psi( P',C''_{f},C'_{-f})(f)=\emptyset $. Hence, (\ref{defENF}) becomes $ \psi(P',C')(f)= C_f\left(\psi(P',C')(f)\right)  $ and $ \psi(P',C') $ is not blocked by any firm under $ (P,C) $. Therefore, $ \psi(P',C')\in \mathcal{I}(P,C) $. 
 
 \item [$\boldsymbol{(ii)}$] We need to prove that for each $ \mu \in \mathcal{I}(P,C) $ there is a Nash equilibrium $  (P^{\star},C^{\star}) $ of $ (\psi,P,C) $ such that $ \psi( P^{\star},C^{\star})=\mu$. In order to do so, given $\mu \in \mathcal{I}(P,C) $, we define $ (P^{\star},C^{\star}) $ in which each firm $ f $ declares $C^{\star}_f$ such that:
 
\begin{itemize}
 \item[$\diamond$] for each $S \subseteq W $, $ C^{\star}_f(S)=\mu(f)\cap S. $
 \item[$\diamond$] $C^{\star}_f(S')\subsetneq C^{\star}_f(S)$ for each  $ S'\subsetneq S\subseteq \mu(f)$
\end{itemize}

and each worker $ w $ declares $P_w^{\star}= \mu(w),\emptyset $. It can be verified that $C^{\star}_f$ is substitutable for each firm $ f\in F. $
Now, we need to prove that (1) $ \psi(P^{\star},C^{\star})=\mu $, and (2) $ (P^{\star},C^{\star}) $ is a Nash equilibrium of the game.
\begin{enumerate}[(1)]
\item \textbf{$\boldsymbol{\psi(P^{\star},C^{\star})=\mu }$}. Since $ \psi $ is a stable matching rule under $ (P^{\star},C^{\star}) $, to prove (1) it is enough to show that   $ S(P^{\star},C^{\star})=\lbrace \mu\rbrace $. First, we consider the DA algorithm \footnote{See \cite{roth1992two}  for more details.} when workers make offers under  $ (P^{\star},C^{\star}) $. On the one hand, each unemployed worker under  $ \mu$ makes no offers in any step of the DA algorithm, that is, it is an unemployed worker under $ \mu_W(P^{\star},C^{\star})$. On the other hand, each worker $ w$ assigned to a firm $ f $ under $ \mu$ makes an offer to the firm $ f $. Since $ C^{\star}_{f}\left(  \mu(f)\right) =\mu(f)$, firm $f$ accepts all offers received. So $ \mu_W(P^{\star},C^{\star})= \mu.$ Second, we consider the DA algorithm when firms make offers under $ (P^{\star},C^{\star}) $.  Notice that each unmatched firm under $ \mu $ makes no offers in any step of the DA algorithm, i.e., it is an unmatched firm under $ \mu_F(P^{\star},C^{\star})$. For each matched firm under $ \mu $, $ C^{\star}_f(W)= \mu(f)$,  so it makes offers to each worker $ w $ in  $ \mu(f) $. Since each worker $ w$ in $ \mu(f) $ has only firm $ f $ as acceptable, she accepts it. So $ \mu_F(P^{\star},C^{\star})= \mu.$ Therefore, $ \mu_F(P^{\star},C^{\star})=\mu_W(P^{\star},C^{\star})=\mu $, i.e. $ \mu $ is the unique stable matching under $(P^{\star},C^{\star})$.

 \item \textbf{$\boldsymbol{(P^{\star},C^{\star})}$} \textbf{is a Nash equilibrium of }$ \boldsymbol{ (\psi,P,C) }.$ Otherwise, there are two cases to consider:

\textbf{Case (2.1):} \textbf{ There is a firm $\boldsymbol{ f\in F} $ and  a choice function $ \boldsymbol{\widehat{C}_f}   $ such that }

$\boldsymbol{ \psi(P^{\star},C^{\star})(f)  \succeq_f \psi( P^{\star},\widehat{C}_{f},C^{\star}_{-f })(f) }$ \textbf{ does not hold.}  The strategy profile $ (\widehat{C}_{f},C^{\star}_{-f }) $ is denoted by $ \widehat{C} $. 
Thus, $\psi(P^{\star},C^{\star}) (f)\neq C_{f}\left( \psi(P^{\star},C^{\star})(f)\cup \psi( P^{\star},\widehat{C})(f)\right) . $
Then, there is 
$ 
w'\in C_{f}\left( \psi(P^{\star},C^{\star})(f)\cup \psi( P^{\star},\widehat{C})(f)\right)$ such that   $ w'\notin  \psi(P^{\star},C^{\star}) (f)=\mu(f).$
Hence, $ w'\in \psi(P^{\star},\widehat{C})(f) $  and $ w'\notin \mu(f). $ Thus, definition of $ P^{\star} $ implies 
 $$ 
 \mu(w')P^{\star}_{w'}\emptyset P^{\star}_{w'}f=\psi(P^{\star},\widehat{C})(w').
 $$ 
  Contradicting $\psi(P^{\star},\widehat{C})\in I(P^{\star},\widehat{C})$.  

\textbf{Case (2.2):} \textbf{ There is a worker $\boldsymbol{ w\in W} $ and  a strategy  $ \boldsymbol{\widehat{P}_w } $ such that\\ $\boldsymbol{\psi(P^{\star},C^{\star})(w)R_w \psi(\widehat{P}_w,P^{\star}_{-w},C^{\star}_{f}) (w)} $ does not hold.}  The strategy profile $ (\widehat{P}_w,P^{\star}_{-w}) $ is denoted by $ \widehat{P} $. Thus,
\begin{equation}\label{desP}
\psi(\widehat{P},C^{\star}_{f})(w)P_w \psi(P^{\star},C^{\star})(w).
\end{equation}

 Let $f' \in F$ such that $f'=\psi(\widehat{P},C^{\star}_{f})(w)$. By (1), $\psi(P^{\star},C^{\star})(w)=\mu(w)$  and (\ref{desP}) becomes $f' P_w \mu(w). $ The fact that $ w\notin \mu(f') $ and the definition of  $ C^{\star}$ imply that 
 $  C^{\star}_{f'}\left(  \psi(\widehat{P},C^{\star})(f')\right)\neq \psi(\widehat{P},C^{\star})(f') $ contradicting $\psi(\widehat{P},C^{\star})\in I(\widehat{P},C^{\star})$. 

\end{enumerate}
By Cases (2.1) and (2.2), $(P^{\star},C^{\star})$ is a Nash equilibrium of $ (\psi,P,C). $
\end{enumerate}

Therefore by \textit{(i)} and \textit{(ii)}, $ (\psi,P,C) $ implements in Nash equilibrium the individually rational correspondence $ \mathcal{I} $.
\end{proof}

The existence of Nash equilibria follows from Theorem \ref{ToremaIR}. A natural question is which is the game that implements stable matchings in Nash equilibrium. The answer to this question is the focus of the next subsection.

 \subsection{The stable correspondence}\label{s2}

Now consider a matching game in which the players are only the workers. Given the profile of choice functions $ C $, a \textbf{social choice correspondence}  $  \Psi^{C}: \mathcal{P}  \rightarrow 2^{\mathcal{M}}$ selects, for each strategy profile $ P'\in \mathcal{P}  $, a subset of matchings $ \Psi^{C}(P')\subseteq \mathcal{M}.  $  The \textbf{stable correspondence} $ \mathcal{S}^{C}:\mathcal{P}  \rightarrow 2^{\mathcal{M}} $ selects for each strategy profile $ P'\in \mathcal{P} $, the set of stable matchings $ S(P',C) $. 
 Given the profile of choice functions $ C $, a \textbf{(matching) rule} is a function  $\psi^{C}:  \mathcal{P} \rightarrow \mathcal{M} $  that selects for each strategy profile $ P' \in \mathcal{P}  $ a matching $ \psi^{C} (P')\in \mathcal{M}. $  A rule $ \psi^{C} $ is \textbf{stable} if $ \psi(P')\in S(P',C) $ for each $ (P',C)\in \mathcal{P}\times \mathcal{C}.$ 
 
Given the profile of choice functions $ C $, \textbf{the firm-optimal stable rule} $ \psi^{C}_F$ is a function  $\psi_F^{C}:   \mathcal{P} \rightarrow \mathcal{M} $  that selects for each strategy profile $ P' \in \mathcal{P}$ the firm-optimal stable matching $ \psi^C_F (P') $  under $ (P',C). $

Given  $  P\in \mathcal{P} $ and the firm-optimal stable rule $ \psi^{C}_F$ the \textbf{(matching) game} induced by $   P $ and $ \psi^{C}_F $ is denoted by $ (\psi^{C}_F, P). $ A strategy profile $P'  $ is a \textbf{Nash equilibrium} of $(\psi^{C}_F, P)$, if no worker can achieve a better outcome deviating from her strategy, assuming that the other workers do not deviate from the strategy profile $P'  $.  Formally,

\begin{definition}
Let be $(\psi^{C}_F, P)$ the game induced by $ P $ and the stable rule $ \psi^{C}_F $. A strategy profile $ P' $ is a \textbf{Nash equilibrium} of $(\psi^{C}_F, P)$  if for each $ w\in W $, $\psi(P',C')(w)R_w\psi( \widehat{P}_w,P'_{-w},C')(w) $ for each strategy $ \widehat{P}_w  $ of $ w $.
\end{definition}

\begin{definition}\label{defimple2}

 We say that the game $(\psi^{C}_F, P)$ \textbf{implements} the social choice correspondence $ \Psi^{C} $ in Nash equilibrium if,
\begin{enumerate}[(i)]
\item     for each Nash equilibrium  $ P' $ of $(\psi^{C}_F, P)$, $ \psi^{C}_F(P')\in  \Psi^{C}(P) $, 

 \item  for each matching $ \mu \in  \Psi^{C}(P) $ there is a Nash equilibrium $ P' $ of the game $(\psi^{C}_F, P)$ such that $ \psi^{C}_F(P')=\mu. $
 \end{enumerate}
 \end{definition}

By means of an additional condition on firms' choice functions, we can implement the stable correspondence in Nash equilibrium. This additional condition is 
the “law of aggregate demand”, which says that when a firm chooses from an expanded set, it hires at least as many workers as before. Formally,

\begin{definition}
Choice function $ C_f $ satisfies \textbf{ law of aggregate demand (LAD)} if $S'\subseteq S\subseteq W$ implies $|C_f(S')|\leq |C_f(S)|.$
\end{definition}

 We denote by $ \mathcal{C}_{LAD}$  the set of all firms' choice functions $ C\in \mathcal{C} $ such that $C_f$ satisfies LAD for each $ f\in F $. The following theorem asserts that the firm-optimal stable rule implements, in Nash equilibrium, the stable correspondence. 

\begin{theorem}\label{S(P)=NE}
Let $ (P,C) $ be a market, $ C\in \mathcal{C}_{LAD}, $ and let  $ \psi^{C}_F $ be the firm-optimal stable matching rule. Then, the game $ (\psi^{C}_F, P) $ implements the stable correspondence $ \mathcal{S}^{C} $ in Nash equilibrium.  
\end{theorem}
\begin{proof}
Let $ P $ be a preference profile, $ C $ be a profile of choice functions and $\psi^{C}_F$ the firm-optimal stable  rule. In order to show that the game $ (\psi^{C}_F, P) $ implements $ \mathcal{S}^{C} $ in Nash equilibrium, we need to prove items \textit{(i)} and \textit{(ii)} of Definition \ref{defimple2}.
\begin{enumerate}
\item [$\boldsymbol{(i)}$]  Let $P'$ be a Nash equilibrium of the game $ (\psi^{C}_F,P) $.  We prove that $\psi^{C}_F(P') \in \mathcal{S}^{C}(P).$  An argument analogous to the proof of Theorem \ref{ToremaIR} \textit{(i)}, shows that $ \psi^{C}_F $ is not blocked by any worker. Since $ \psi^{C}_F $ is a stable  rule under $ (P,C) $, $ \psi^{C}_F $ is not blocked by any firm.  Therefore $ \psi^{C}_F(P')\in I(P,C). $
Assume that $ \psi^{C}_F (P') \notin S(P,C)$. Thus, there is a blocking pair $(f,w) $ for  $ \psi^{C}_F (P'),$  
\begin{equation}\label{bloqueoh(Q)}
 w\in C_f \left(  \psi^{C}_F (P')(f)\cup\left\lbrace w\right\rbrace  \right) ~\text{and}~  fP_w  \psi^{C}_F (P')(w). 
\end{equation}

Now, consider the strategy profile  $ P''=(P''_{w},P'_{-w}) $ where $ P''_w=f,\emptyset $ . We claim that $  \psi^{C}_F (P'')(w)=f$. Otherwise, since $\psi^{C}_F (P'')\in I(P'',C)$, $ \psi^{C}_F (P'')(w)=\emptyset $.  The definition of  $ P''$ implies that
$fP''_{w} \psi^{C}_F (P'')(w)=\emptyset.$ If  $w\in C_{f} \left( \psi^{C}_F (P'')(f)\cup \left\lbrace w\right\rbrace\right) ,$ then $ (f,w) $ blocks $ \psi^{C}_F (P'') $ under $ (P'',C) $,
contradicting that $\psi^{C}_F (P'')\in S(P'',C)$. Thus, $w\notin C_{f} \left(\psi^{C}_F (P'')(f)\cup \left\lbrace w\right\rbrace\right) .$  Then, by substitutability, \\
$ w\notin C_{f} \left( \psi^{C}_F (P'')(f)\cup \psi^{C}_F (P')(f) \right)$. Applying Theorem 1 in  \citet{CRAWFORD1991389},\footnote{The version of this theorem for a many-to-one matching market, where all firms have substitutable choice functions, states that by removing one or more workers from the market, the firm-optimal stable matching rule in the new reduced market is weakly worse for all firms than the firm-optimal stable matching rule in the original matching market. In this case, we consider the reduced market $ (F,W\setminus\lbrace w\rbrace, P'_{-w}) $.} 
 \begin{equation}\label{b3} 
w\notin C_{f} \left( \psi^{C}_F (P'')(f)\cup \psi^{C}_F (P')(f) \right)=\psi^{C}_F (P')(f).
\end{equation}

Now, using (\ref{propiedad de choice}) and \eqref{b3} we get
$$
w\notin C_{f} \left( \psi^{C}_F (P'')(f)\cup \psi^{C}_F (P')(f) \right)=C_{f} \left( C_{f} \left( \psi^{C}_F (P'')(f)\cup  \psi^{C}_F (P')(f) \right)\cup \left\lbrace w\right\rbrace\right)
$$ 
 $$= C_f \left( \psi^{C}_F (P')(f)\cup\left\lbrace w\right\rbrace  \right).
 $$
 Thus, $ w\notin C_f \left( \psi^{C}_F (P')(f)\cup\left\lbrace w\right\rbrace  \right) $, but this contradicts \eqref{bloqueoh(Q)} and the claim is proven. Therefore, $\psi^{C}_F (P'')(w)=f .$ By (\ref{bloqueoh(Q)}), we have $  \psi^{C}_F (P'')(w)P_w \psi^{C}_F (P')(w)$, then $P' $ is not a Nash equilibrium, contradicting our hypothesis. Therefore $\psi^{C}_F (P')\in \mathcal{S}^{C}(P). $

 \item [$\boldsymbol{(ii)}$] We need to prove that for each $ \mu \in \mathcal{S}^{C}(P) $ there is a Nash equilibrium $  P^{\star} $ of $ (\psi^{C}_F,P) $ such that $ \psi^{C}_F( P^{\star})=\mu$.  In order to do so, given $\mu \in \mathcal{S}^{C}(P) $ we define  $ P^{\star}$ in which each worker $ w $ declares $P^{\star}_w= \mu(w),\emptyset $.  Now, we need to prove that (1) $ \psi^{C}_F(P^{\star})=\mu $, and (2) $ P^{\star} $ is a Nash equilibrium of the game.
\begin{enumerate}[(1)]
\item \textbf{$\boldsymbol{\psi^{C}_F(P^{\star})=\mu }$}. Since $ \psi^{C}_F $ is a stable  matching rule under $ (P^{\star},C) $, to prove (1) is suffices to show that   $ S(P^{\star},C)=\lbrace \mu\rbrace $. 

Following similar reasoning to the ones in the proof of Theorem \ref{ToremaIR}, the item \textit{(ii)} part (1), $ \mu_W(P^{\star})=\mu $. 
 Assume that $  \mu_W(P^{\star})\neq \mu_F(P^{\star})$, then there is $ w\in W $ such that $ \mu_W(P^{\star})(w)\neq \mu_F(P^{\star})(w). $ Since $ \mu(w) $ is the only acceptable firm under $ P^{\star} $, $  \mu_F(P^{\star})(w)=\emptyset $, but this contradicts the Rural Hospital Theorem \footnote{The \textit{Rural Hospital Theorem} is proven in different contexts by many authors \citep[see][among others]{mcvitie1970stable,roth1984evolution,roth1985college,martinez2000single,alkan2002class,kojima2012rural}.  The version of this theorem for a many-to-many matching market where all agents have substitutable choice functions satisfying \emph{LAD}, that also applies in our setting, is presented in \citet{alkan2002class} and states that each agent is matched with the same number of partners in every stable matching.}. Therefore $ \mu_W(P^{\star})=\mu_F(P^{\star})=\mu, $ i.e. $ \mu $ is the unique stable matching under $ (P^{\star},C) $.

\item \textbf{$\boldsymbol{P^{\star}}$ is a Nash equilibrium of} $ \boldsymbol{(\psi^{C}_F,P).} $ Otherwise, there is a worker $\boldsymbol{ w\in W} $ and  a strategy  $ \boldsymbol{\widehat{P}_w } $ such that $\boldsymbol{\psi^{C}_F(P^{\star})(w)R_w \psi_F(\widehat{P}_w,P^{\star}_{-w})(w) } $ does not hold. The strategy profile $(\widehat{P}_w,P^{\star}_{-w}) $ is denoted by $ \widehat{P}.$ Thus,

\begin{equation}\label{sino}
\psi^{C}_F(\widehat{P})(w)P_w \psi^{C}_F(P^{\star})(w).
\end{equation}
Since $\psi^{C}_F(P^{\star})(w)=\mu(w)$ and $ \mu \in I(P,C)$, there is $f' \in F$ such that $f'=\psi^{C}_F(\widehat{P})(w)$. If $ w\in C_{f'}\left( \psi^{C}_F (P^{\star})(f')\cup\lbrace w\rbrace\right) $, this together with \eqref{sino} imply that $ (f',w) $ blocks $ \mu $ under $ (P,C) $. This is a contradiction.  Therefore, 

\begin{equation}\label{wnocho}
 w\notin C_{f'}\left( \psi^{C}_F(P^{\star})(f')\cup \lbrace w\rbrace \right) .
 \end{equation}
We claim that 
\begin{equation}\label{inc}
\psi^{C}_F(\widehat{P})(f')\subsetneq \psi^{C}_F(P^{\star})(f')\cup \lbrace w\rbrace .
\end{equation}

 Let $ \widehat{w} \neq w$ such that $ \widehat{w}\in \psi^{C}_F(\widehat{P})(f') $. Definition of $ \widehat{P} $ and $ \psi^{C}_F (\widehat{P})\in I(\widehat{P})$ imply that $f'= \psi^{C}_F (\widehat{P})(\widehat{w})=\mu(\widehat{w})=\psi^{C}_F(P^{\star})(\widehat{w}). $ Thus, $ \widehat{w}\in  \psi^{C}_F(P^{\star})(f').$ Now, assume that 
 
\begin{equation}\label{igualdad}
 \psi^{C}_F(\widehat{P})(f')=\psi^{C}_F(P^{\star})(f')\cup \lbrace w\rbrace .
 \end{equation}
 
 Using, $ \psi^{C}_F(\widehat{P})\in I(\widehat{P},C)$, definition of $\widehat{P}$ and \eqref{igualdad}, it follows that 
$$
\psi^{C}_F(\widehat{P})(f')=C_{f'}\left(\psi^{C}_F(\widehat{P})(f')\right) =C_{f'}\left(\psi^{C}_F(P^{{\star}})(f')\cup \lbrace w\rbrace\right).
$$
  Using, \eqref{wnocho} and $ \psi^{C}_F(P^{\star})\in I(P^{\star},C) $, it follows that 
$$
 C_{f'}\left( \psi^{C}_F(P^{\star})(f')\cup \lbrace w\rbrace\right)=C_{f'}\left( \psi^{C}_F(P^{\star})(f')\right)=\psi^{C}_F(P^{\star})(f').
$$

Hence, $\psi^{C}_F(\widehat{P})(f')=\psi^{C}_F(P^{\star})(f') $ contradicting \eqref{sino}. This completes the proof of the Claim. Therefore, there is $ w^{\star} \in \psi^{C}_F(P^{\star})(f')\setminus \psi^{C}_F(\widehat{P})(f')$. If
  $ w^{\star}\notin C_{f'}\left( \psi^{C}_F(\widehat{P})(f')\cup \lbrace w^{\star}\rbrace \right) $. Using \eqref{inc}, substitutability and definition of $ \widehat{P}$ ,  $w^{\star}\notin C_{f'}\left( \psi^{C}_F(P^{\star})(f')\cup \lbrace w\rbrace \right)$. Using \eqref{wnocho} and $ \psi^{C}_F(P^{\star})\in I(P^{\star},C) $ we have that, 
  $$ 
w^{\star}\notin C_{f'}\left( \psi^{C}_F(P^{\star})(f')\cup \lbrace w\rbrace \right)=C_{f'}\left( \psi^{C}_F(P^{\star})(f') \right)=\psi^{C}_F(P^{\star})(f'). 
  $$
 Thus, $ w^{\star}\notin \psi^{C}_F(P^{\star})(f')$ contradicting that $ w^{\star} \in \psi^{C}_F(P^{\star})(f')\setminus \psi^{C}_F(\widehat{P})(f')$. Then, 
 \begin{equation}\label{blok}
 w^{\star}\in C_{f'}\left( \psi^{C}_F( \widehat{P})(f')\cup \lbrace w^{\star}\rbrace \right).
\end{equation} 
The fact that $ \mu(w^{\star})=\psi^{C}_F(P^{\star})(w^{\star})\neq \psi^{C}_F(\widehat{P})(w^{\star}) $ and the definition of $ \widehat{P}$ imply that, $f'=\mu(w^{\star})P^{\star}_{w^{\star}}\psi^{C}_F(\widehat{P})(w^{\star})$. This together with  \eqref{blok} imply that $ (f',w^{\star}) $ blocks $ \psi^{C}_F(\widehat{P}) $ under $ (\widehat{P},C)$, contradicting the stability of  $ \psi^{C}_F(\widehat{P}) $ under $ (\widehat{P},C)$. 

Therefore, $P^{\star}$ is a Nash equilibrium.

\end{enumerate}
\end{enumerate}
Finally, by \textit{(i)} and \textit{(ii)}, $ (\psi^{C}_F,P) $ implements in Nash equilibrium the stable correspondence $ \mathcal{S}^{C} $.
\end{proof}

\section{Concluding Remarks}\label{Concludings}

The main motivation of this paper is to provide a framework to study the Nash equilibrium solutions of the game induced by stable rules. In a many-to-one matching market with substitutable choice functions, we show that any stable matching rule implements, in Nash equilibrium, the individually rational matchings. Moreover, when only workers play strategically and firms' choice functions satisfy the law of aggregate demand, we show that the firm-optimal stable rule $ \psi^{C}_F $ implements the stable correspondence in Nash equilibrium. The analogous result with workers telling the truth and firms acting strategically does not hold. That is, if we consider a game in which the players are only the firms, we cannot implement the stable correspondence in Nash equilibrium.  This fact was already noted by \cite{Sotomayor2012}, even under a more restrictive model with responsive preferences.

It is usual in the literature to study many-to-one markets assuming that firms’ preferences are responsive. This is due to the close relationship between this market with responsive preferences and the marriage market , \citep [see]  [for a thorough survey on this fact]{roth1992two} . However, when firms are endowed with substitutable preferences (a much less restrictive requirement), this relation with the marriage market no longer holds.  Thus, extending the results of a many-to-one market with responsive preferences to substitutable choice functions is not straightforward.

The study of the implementability of several solution concepts under other equilibria notions is an interesting topic for future research.


\end{document}